\documentclass[letterpaper, 10 pt, conference]{ieeeconf}   

\IEEEoverridecommandlockouts
\usepackage{hyperref}
\usepackage{graphicx}
\usepackage[utf8]{inputenc}
\usepackage{cite}
\usepackage{multicol}
\usepackage{mathtools, cuted}
\usepackage[cmintegrals]{newtxmath}
\usepackage{mathtools}

\usepackage{subcaption}
\usepackage{epsfig}
\usepackage{amsmath}
\usepackage{amssymb}
\usepackage{algorithmic}

\newtheorem{theorem}{Theorem}
\usepackage{caption}
\usepackage[usenames, dvipsnames]{color}

\title{\LARGE \bf Experimental Evaluation of Continuum Deformation\\ 
with a Five Quadrotor Team}
\author{Matthew Romano$^{*,\dagger,a}$, Prince Kuevor$^{\dagger,a}$, Derek Lukacs$^{b}$, Owen Marshall$^a$, Mia Stevens$^{a}$,\\
Hossein Rastgoftar$^{b}$, James Cutler$^b$, Ella Atkins$^{a,b}$%
\thanks{* Corresponding author}
\thanks{$^{\dagger}$ Both authors contributed equally to this work.}%
\thanks{$^a$Authors are with the Robotics Institute at the University of Michigan {\tt\small\{mmroma, kuevpr, oamarsh, minist, ematkins\}@umich.edu}}%
\thanks{$^b$Authors are with the Department of Aerospace Engineering at the University of Michigan {\tt\small\{dlukacs, hosseinr, jwcutler, ematkins\}@umich.edu}}%
}
\date{August 2018}

\begin{document}

\maketitle

\begin{abstract}
This paper experimentally evaluates continuum deformation cooperative control for the first time.  Theoretical results are expanded to place a bounding triangle on the leader-follower system such that the team is contained despite nontrivial tracking error. Flight tests were conducted with custom quadrotors running a modified version of ArduPilot on a BeagleBone Blue in M-Air, an outdoor netted flight facility. Motion capture and an onboard inertial measurement unit were used for state estimation. Position error was characterized in single vehicle tests using quintic spline trajectories and different reference velocities. Five-quadrotor leader trajectories were generated, and followers executed the continuum deformation control law in-flight. Flight tests successfully demonstrated continuum deformation; future work in characterizing error propagation from leaders to followers is discussed.
\end{abstract}

\section{Introduction}

Cooperative control is a popular area of theoretical research.   
Virtual structure (VS) \cite{ren2004decentralized}, consensus \cite{ren2007information, rahbari2014incremental, inzucchi2015management}, containment control  \cite{qin2017containment, li2015containment}, and continuum deformation \cite{rastgoftar2016continuum, rastgoftar2017continuum} are examples of multi-agent system  (MAS) control. VS is a centralized approach, while others are decentralized. 
Consensus is the most commonly-applied decentralized cooperative control technique \cite{liu2015finite, cao2013overview, rahbari2014incremental, yang2016distributed, inzucchi2015management}. Distributed consensus was applied for agent coordination in \cite{rao2014sliding, ren2009distributed} and flight tested in \cite{han2014multiple, guerrero2012mini}. Consensus guided by a single leader is studied in \cite{song2010second, xu2016robust} and flight tested in \cite{drak2014autonomous, namerikawa2018consensus}.
Cooperative control has been applied to unmanned aircraft system (UAS) teams for tasks such as surveillance \cite{nigam2012control}, area surveys \cite{han2013low}, and payload delivery \cite{rastgoftar2018cooperative}. 

\begin{figure}[ht!]
  \centering
    \includegraphics[width=\columnwidth,trim=363 159 420 253, clip]{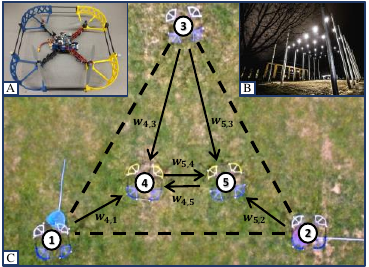}
  \caption{A) Test Quadrotor, B) M-Air Netted Facility, C) Mid-flight, overhead snapshot in M-Air with three leaders (1,2,3) and two followers (4,5) flying a continuum deformation.}
  \label{fig:quadrotor_system}
\end{figure}

Continuum deformation is a multi-agent control technique that allows translation, rotation, and shearing of a bounding envelope while ensuring agents remain within the bounding envelope and avoid collisions. 
To do this, continuum deformation treats agents as particles of a body that deforms under a homogeneous transformation \cite{rastgoftar2016continuum, rastgoftar2014evolution}. A desired $n$-D homogeneous transformation ($n=1,2,3$) is defined by $n+1$ leaders in $\mathbb{R}^n$ and acquired by followers through local communication. Shared state data is weighted based on a reference configuration.
Continuum deformation stability is analyzed in \cite{rastgoftar2016continuum}, while coordination under switching communication topologies is studied in \cite{rastgoftar2017continuum}.  Characterization of the homogeneous transformation and safety guarantees are key features of continuum deformation.  

This paper presents results from the first experimental evaluation of continuum deformation. Single-agent error is characterized and used in designing leader agent trajectories that satisfy four constraints: (1) Follower containment, (2) Collision avoidance, (3) Bounding of agent deviation from local desired position, and (4) Bounding of agent deviation from global desired position. A five-quadrotor team receiving motion capture data is deployed in tests. Leaders execute prescribed trajectories; followers receive neighbor positions via the communication topology in Fig. \ref{fig:quadrotor_system}C. Tracking errors of leaders and followers are analyzed with respect to the bounding envelope designed to contain the team given expected single-vehicle deviations. 

Sec. \ref{sec:continuum-deformation} summarizes continuum deformation theory, while Sec. \ref{sec:generating-path} describes leader flight planning and presents test trajectories. Sec. \ref{sec:experimental-setup} summarizes the experimental apparatus including quadrotors, electronics, sensors, and software.  Sec. \ref{sec:experimental-results} presents flight test results, followed by a discussion (Sec. \ref{sec:discussion}) and conclusion (Sec. \ref{sec:conclusion}).





\section{Continuum Deformation Coordination Review}
\label{sec:continuum-deformation}

A $2D$ continuum deformation is defined by a homogeneous transformation  \cite{lai2009introduction}:
\begin{align}
\label{Homogggennooustransformmm}
t\geq t_0, \qquad  \mathbf{r}_{i,HT}(t)=\mathbf{Q}(t,t_0)\mathbf{r}_{i,0}+\mathbf{d}(t,t_0).
\end{align}
where $t_0$ denotes the initial time,
$\mathbf{Q}(t,t_0) \in \mathbb{R}^{2 \times 2}$
is a non-singular planar deformation matrix,
$\mathbf{d}(t,t_0) \in \mathbb{R}^{2 \times 1}$
is the rigid body displacement vector, and $\mathbf{r}_{i,HT}(t)=[x_{i,HT}(t)~y_{i,HT}(t)]^T$ is the \emph{global desired position} of agent $i$. Note that $\mathbf{d}(t_0,t_0)=\mathbf 0 \in \mathbb R^{2 \times 1}$, and $\mathbf{Q}(t_0,t_0)= \mathbf{I}_2 \in \mathbb R^{2 \times 2}$, so $ \mathbf{r}_{i,0}=\mathbf{r}_{i,HT}(t_0)=[x_{i,0}~y_{i,0}]^T$ is the initial position of quadrotor $i\in \mathcal{V}$. For the remainder of this paper, the argument of time may be ommited from state variables for brevity, i.e. $\mathbf{r}_{i,HT} = \mathbf{r}_{i,HT}(t).$
\\
\textbf{Continuum Deformation Definition:} A 2D homogeneous transformation can be acquired by an $N$-agent quadrotor team with index numbers $\mathcal{V}=\{1,\cdots,N\}$. Let $\mathcal{V}=\mathcal{V}_L\bigcup \mathcal{V}_F$, with leaders $\mathcal{V}_L=\{1,2,3\}$ and followers $\mathcal{V}_F=\{4,\cdots,N\}$. Leaders form a \textit{leading triangle} for all $t\geq t_0$. Followers acquire the continuum deformation by local communication. Leader positions can be uniquely related to $\mathbf{Q}$ and $\mathbf{d}$ \cite{rastgoftar2016continuum}:
\begin{equation}
t\geq t_0,\qquad 
\begin{bmatrix}
\mathrm{vec}\left(\mathbf{Q}^T\right)\\
\mathbf{d}
\end{bmatrix}
=\begin{bmatrix}
\mathbf{I}_2\otimes\mathbf{P}_0&\mathbf{I}_2\otimes\mathbf{1}_3
\end{bmatrix}
^{-1}
\mathrm{vec}\left(\mathbf{P}_{HT}\right),
\end{equation}
where "$\otimes$" is the Kronecker product symbol, $\mathbf{I}_2\in \mathbb{R}^{2\times 2}$ is the identity matrix, and $\mathbf{1}_3\in \mathbb{R}^{3\times 1}$ is a vector of ones,
\[
\mathbf{P}_{0}=
\begin{bmatrix}
\mathbf{r}_{1,0}^T\\
\mathbf{r}_{2,0}^T\\
\mathbf{r}_{3,0}^T\\
\end{bmatrix}
\in \mathbb{R}^{3\times 2},~\mathbf{P}_{HT}=\begin{bmatrix}
\mathbf{r}_{1,HT}^T\\
\mathbf{r}_{2,HT}^T\\
\mathbf{r}_{3,HT}^T\\
\end{bmatrix}
\in \mathbb{R}^{3\times 2}
.
\]
\textbf{Continuum Deformation Acquisition:} A weighted directed graph $\mathcal{G}=\mathcal{G}\left(\mathcal{V},\mathcal{E}\right)$ defines inter-agent communication as shown in Fig. \ref{fig:quadrotor_system}C. Given edge set $\mathcal{E}\subset \mathcal{V}\times \mathcal{V}$, the in-neighbor agents of follower $i\in \mathcal{V}_F$ is defined by 
\[
\forall i\in \mathcal{V}_F,\qquad \mathcal{N}_i=\{j|(i,j)\in \mathcal{E}\}.
\]
For an $n$-D continuum deformation, $\big|\mathcal{N}_i\big|=n+1,~\forall i\in \mathcal{V}_F$ and in-neighbor agents of every follower form an $n$-D polytope. Therefore, each follower communicates with three in-neighbor agents, forming a triangle in a 2D continuum deformation. Let $i_1$, $i_2$, and $i_3$ denote index numbers of follower $i$'s in-neighbor agents, initially positioned at $\mathbf{r}_{i_1,0}$, $\mathbf{r}_{i_2,0}$, and $\mathbf{r}_{i_3,0}$. Then, communication weights of follower $i\in \mathcal{V}_F$ denoted $w_{i,i_1}$, $w_{i,i_2}$, and $w_{i,i_3}$ are given by \cite{rastgoftar2016continuum}:
\begin{equation}
\begin{bmatrix}
w_{i,i_1}\\
w_{i,i_2}\\
w_{i,i_3}
\end{bmatrix}
=
\begin{bmatrix}
x_{i_1,0}&x_{i_2,0}&x_{i_3,0}\\
y_{i_1,0}&y_{i_2,0}&y_{i_3,0}\\
1&1&1
\end{bmatrix}
^{-1}
\begin{bmatrix}
x_{i,0}\\
y_{i,0}\\
1
\end{bmatrix}
.
\end{equation}
Any edge that does not appear in the in-neighbor set of $\mathcal{N}_i$ of any follower $i \in \mathcal{V}_F$ is zero. Let $\mathbf{r}_i$ denote actual position of agent $i\in \mathcal{V}$. The \emph{local desired position} of agent $i\in \mathcal{V}$ is:
\begin{equation}
\label{eqn:localDesired}
\mathbf{r}_{i,d}=
\begin{cases}
\mathbf{r}_{i,HT}&i\in \mathcal{V}_L\\
\mathbf\sum_{j\in \mathcal{N}_i} w_{i,j}\mathbf{r}_j&i\in \mathcal{V}_F
\end{cases}
\end{equation}
where $\mathbf{r}_{i,d}$ is the reference trajectory of each quadrotor $i\in \mathcal{V}$. Thus, local and global desired positions are the same for leaders, but this need not hold for followers. 
\\
\textbf{Continuum Deformation Coordination Safety:} Given leader desired trajectories, we define a bounding triangle as a dilation of the original leading triangle with the following properties: (1) Both leading and bounding triangles have a common centroid at any time $t$, (2) Parallel sides of both triangles are separated by $D_l$ at any time $t$ (see Fig. \ref{fig:bothArr}).

\begin{figure}[]
 	\centering
 	\begin{subfigure}[c]{0.73\columnwidth}
	    \includegraphics[width = \textwidth,trim=110 420 135 77, clip]{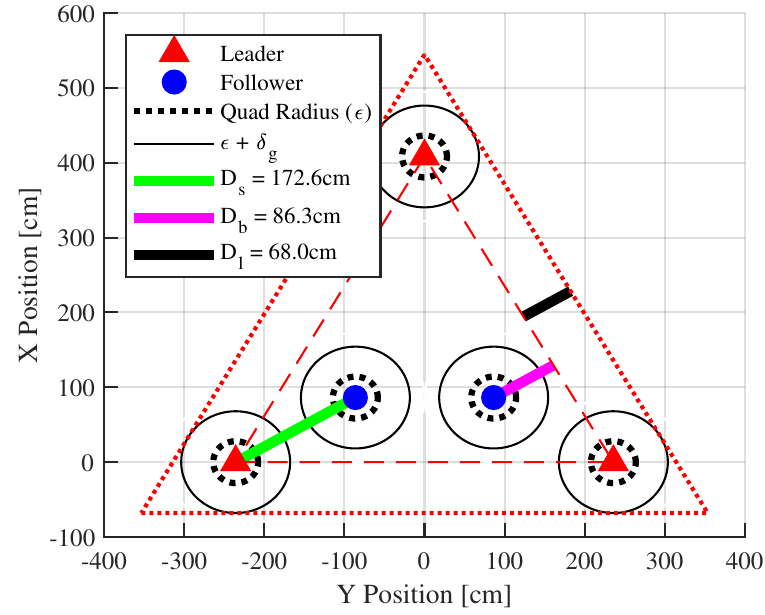}
	    \caption{Initial configuration.} 
	    \label{fig:initial-config}
    \end{subfigure}
    \vspace{0.2cm}
 	\begin{subfigure}[c]{0.73\columnwidth}
	    \includegraphics[width = \textwidth,trim=110 72 135 400, clip]{figures/bothArrangements_dilatedDl_edited.pdf}
	    \caption{Contracted configuration.} 
	    \label{fig:contracted-config}
    \end{subfigure}
	\caption{Leaders and followers shown in their global desired positions for initial and contracted configurations, respectively, with leading (smaller) and bounding (larger) triangles shown. $D_s$, $D_b$, and $D_l$ are shown in the initial configuration.} 
	\label{fig:bothArr}
\end{figure}

\begin{theorem}
\label{thm:Thm1}
Assume each quadrotor is enclosed by a ball of radius $\epsilon$. $D_s$ denotes the minimum separation distance between any pair of agents at $t=t_0$, $D_b$ is the minimum distance from any follower to the leading triangle boundary, and $D_l=\delta_{g}+\epsilon$ is the distance of two parallel sides of the leading and bounding triangles. Here, $\delta_{g}$ is the maximum deviation from the \emph{global} desired position for every quadrotor: $\|\mathbf{r}_i-\mathbf{r}_{i,HT}\|\leq\delta_{g}, \forall i\in \mathcal{V} $. Let
 \begin{equation}
    \label{eqn:deltaMax}
    \delta_{g,\mathrm{max}} =\min\{ \frac{D_s - 2\epsilon}{2}, D_b - \epsilon\}.
\end{equation}
Collision avoidance between agents, follower containment within the leading triangle, and leader containment within the bounding triangle are all guaranteed if
\begin{equation}
\label{IMP1}
    \lambda_{\mathrm{min}}\leq \inf\limits_{\forall t}\big\{\lambda_1\left(t\right),\lambda_2\left(t\right)\big\},
\end{equation}
where
\begin{equation}
    \label{eqn:lamMin}
    \lambda_{\mathrm{min}}=\dfrac{\delta_{g}+\epsilon}{\delta_{g,\mathrm{max}}+\epsilon},
\end{equation}
$\lambda_1$ and $\lambda_2$ denote eigenvalues of matrix $\mathbf{U}_D=\left(\mathbf{Q}^T\mathbf{Q}\right)^{\frac{1}{2}}$. 
\end{theorem}
\begin{proof}
See the proof in \cite{rastgoftar2016asymptotic}.
\end{proof}

This paper implements a 2D, five-agent continuum deformation where $N=5$ and $\mathcal{V}_F=\{4,5\}$. A block diagram of this configuration can be seen in Fig. \ref{fig:Controls_Block_Diagram}. 
The position vectors $\mathbf{r}_{L}, \mathbf{r}_{F}, \mathbf{r}_{L,d}$, and $\mathbf{r}_{F,d}$ denote the leader positions, follower positions, leader desired positions, and follower desired positions respectively and are defined as 

\begin{equation*}
    \mathbf{r}_{L} = 
    \text{vec}(
    \begin{bmatrix}
        \mathbf{r}_{1}^T \\
        \mathbf{r}_{2}^T \\ 
        \mathbf{r}_{3}^T \\
    \end{bmatrix}
    )
    \in \mathbb{R}^{6 \times 1}
    ,~
    \mathbf{r}_{F} = 
    \text{vec}(
    \begin{bmatrix}
        \mathbf{r}_{4}^T \\
        \mathbf{r}_{5}^T \\
    \end{bmatrix}
    )
    \in \mathbb{R}^{4 \times 1},
\end{equation*}

\begin{equation*}
    \mathbf{r}_{L,d} = 
    \text{vec}(
    \begin{bmatrix}
        \mathbf{r}_{1,d}^T \\
        \mathbf{r}_{2,d}^T \\ 
        \mathbf{r}_{3,d}^T \\
    \end{bmatrix}
    )
    \in \mathbb{R}^{6 \times 1}
    ,~
    \mathbf{r}_{F,d} = 
    \text{vec}(
    \begin{bmatrix}
        \mathbf{r}_{4,d}^T \\
        \mathbf{r}_{5,d}^T \\
    \end{bmatrix}
    )
    \in \mathbb{R}^{4 \times 1}.
\end{equation*}

The communication weight matrices $\mathbf{W}_{a}$ and $\mathbf{W}_{b}$ are defined as
\begin{equation*}
    \mathbf{W}_{a} = 
    \begin{bmatrix}
        w_{4,1} & w_{4,2} & w_{4,3} \\
        w_{5,1} & w_{5,2} & w_{5,3} \\
    \end{bmatrix}   
    ,~
    \mathbf{W}_{b} = 
    \begin{bmatrix}
        w_{4,4} & w_{4,5}\\
        w_{5,4} & w_{5,5}\\
    \end{bmatrix}
\end{equation*}

The state variables, $\mathbf{X}_{L}$ and $\mathbf{X}_{F}$, denote the position, velocity, orientation, and angular velocity of each agent. $G_{L}$ and $G_{F}$ represent some unknown dynamics that map the state and motor input commands, $\mathbf{U_L}$ and $\mathbf{U_F}$, to the state dynamics. $H_{L}$ and $H_{F}$ output the respective position vectors from the state vectors $\mathbf{X}_{L}$ and $\mathbf{X}_{F}$.

\begin{figure}[]
 	\centering
	\includegraphics[width = \columnwidth, trim=25 42 0 60, clip]{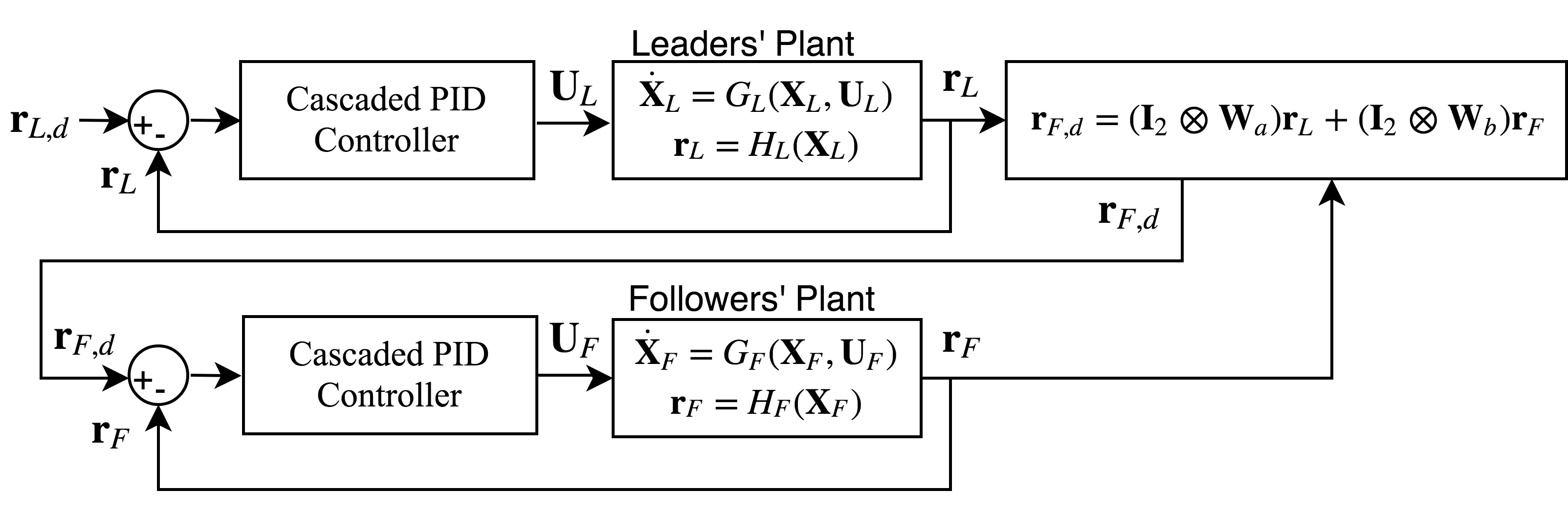}
	\caption{Continuum deformation block diagram.} 
	\label{fig:Controls_Block_Diagram}
\end{figure}

\section{Generating the Leading Triangle Path}
\label{sec:generating-path}
Flight tests investigate if real-world continuum deformation will satisfy four constraints: follower containment within the leading triangle, collision avoidance between agents, bounding of deviation from local desired position, and bounding of deviation from global desired position. 
To evaluate, leaders are set as close together as they theoretically can be.
Leader separation is given by scaling factor $\lambda_{\mathrm{min}}$ (Eq. \ref{eqn:lamMin}). With $\epsilon$ = 28cm (Sec. \ref{sec:experimental-setup}) and $\delta_{g}$ = 40cm (Sec. \ref{sec:experimental-results}), the leaders form an equilateral triangle with edge length $l$ = 3.72m at the contracted $\lambda_{\mathrm{min}}$ limit.

Fig. \ref{fig:trajOutline} shows the planned trajectory by illustrating the path of the leaders.
The flight begins with the leaders contracting from an initial state to the $\lambda_{\mathrm{min}}$ limit, then continues with the leaders traversing a square with edge-length 1m. The flight concludes with leaders expanding back to the initial configuration; the pilot then commands simultaneous descent and landing. During the flight, the leaders fly to six waypoints. Leaders move from their initial equilateral triangle (pose 1) with edge length $l$ = 4.72m to their contracted configuration (pose 2) with $l$ = 3.72m. Leaders then traverse three of the four sides of the 1m edge-length square through poses 3, 4, and 5. From pose 5, the leaders complete the square by returning to pose 2 and finally expand to pose 1 to end in their initial configuration. 

\begin{figure}[]
 	\centering
	\includegraphics[width = \columnwidth,trim=0 140 0 100, clip]{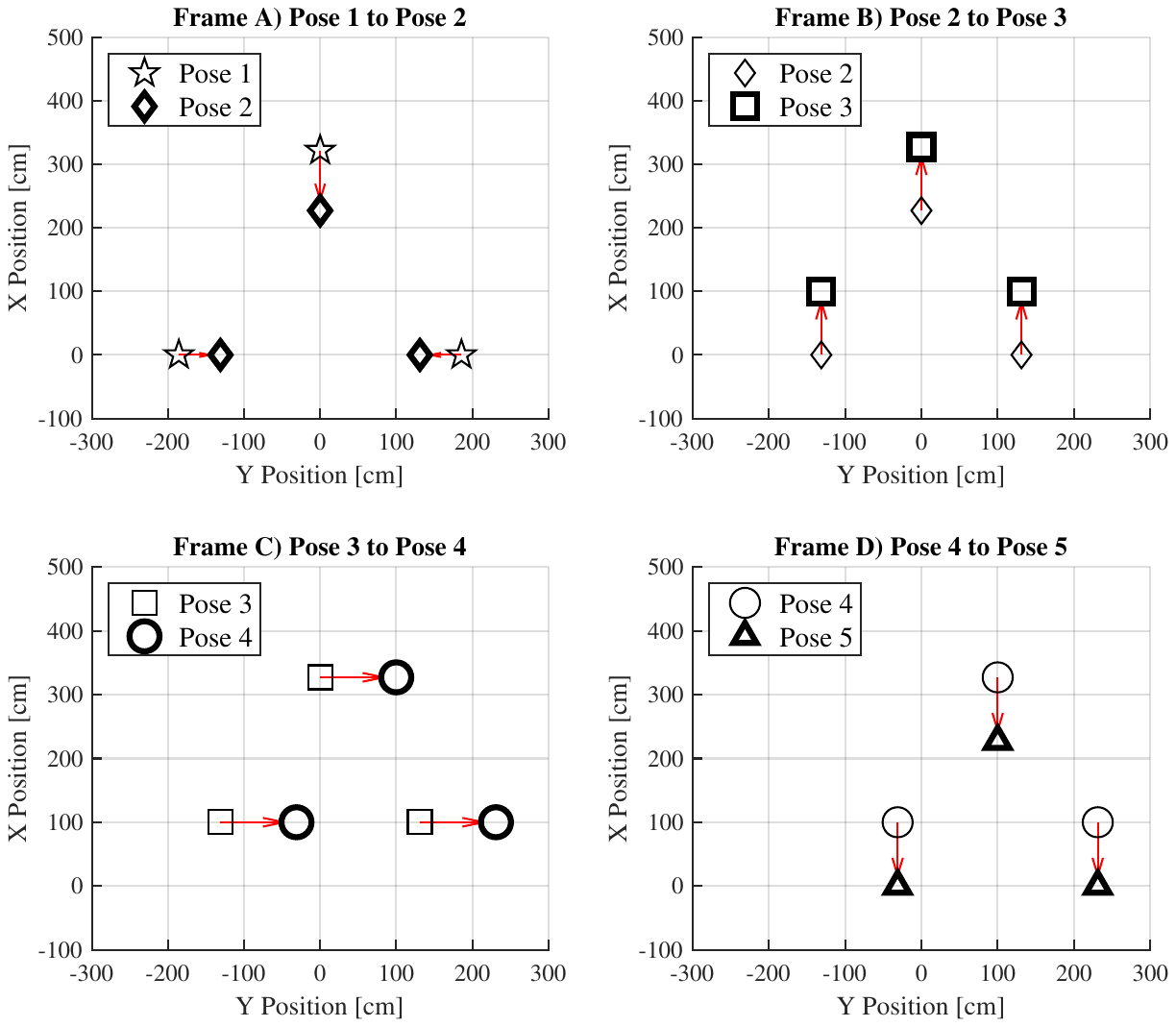}
	\caption{Waypoints traversed by leaders during flight. Poses 1-2 show motion to the contracted position. Poses 2-5 show leaders traversing three sides of the 1m edge-length square. From pose 5, leaders return to pose 2 then to pose 1.}
	\label{fig:trajOutline}
\end{figure}




All leaders move in straight lines to reach their next waypoint. The following quintic spline guidance law is used for generating the desired trajectory for each leader:
\begin{equation}
\label{eqn:spline}
    \begin{split}
        i\in \mathcal{V}_{L},\qquad \mathbf{r}_{i,d}(t)  & = \mathbf{a}_i +  \mathbf{b}_it + \mathbf{c}_it^2 + \mathbf{d}_it^3 + \mathbf{e}_it^4 + \mathbf{f}_it^5\\
        \dot{\mathbf{r}}_{i,d}(t)  & =  \mathbf{b}_i + 2\mathbf{c}_it + 3\mathbf{d}_it^2 + 4\mathbf{e}_it^3 + 5\mathbf{f}_it^4\\
        \Ddot{\mathbf{r}}_{i,d}(t) &= 2\mathbf{c}_i + 6\mathbf{d}_it + 12\mathbf{e}_it^2 + 20\mathbf{f}_it^3,
    \end{split}
\end{equation}

\noindent where $\mathbf{r}_{i,d}(t)$, $\dot{\mathbf{r}}_{i,d}(t)$, and $\Ddot{\mathbf{r}}_{i,d}(t)$ are the 2D desired position, velocity, and acceleration of the $i^{th}$ agent, respectively.  With Eq. \ref{eqn:spline} we enforce six constraints: $\mathbf{r}_{i,d}(t_0) = \mathbf{r}_{i,0}$, $\mathbf{r}_{i,d}(t_f) = \mathbf{r}_{i,f}$, $\dot{\mathbf{r}}_{i,d}(t_0) = \dot{\mathbf{r}}_{i,d}(t_f) = 0$, and $\Ddot{\mathbf{r}}_{i,d}(t_0) = \Ddot{\mathbf{r}}_{i,d}(t_f) = 0$ to solve for coefficients $\{\mathbf{a}_i,\mathbf{b}_i,\mathbf{c}_i,\mathbf{d}_i,\mathbf{e}_i,\mathbf{f}_i\}$ where $\mathbf{r}_{i,0}$ and $\mathbf{r}_{i,f}$ are the initial and final positions of the $i^{th}$ agent, respectively. From this guidance law, we define $v_\mathrm{max}$ as the largest desired velocity of any agent $\dot{\mathbf{r}}_{i,d}(t)$ for all flight times $t$. This desired velocity serves as a feed-forward term in the velocity-tracking portion of the cascaded PID controller (Sec. \ref{sec:experimental-setup}). Time of flight $(t_f - t_0)$ between each pair of poses is 3.75s resulting in a translation of 1m when $v_\mathrm{max}$ = 50$\frac{cm}{s}$.




\section{Experimental Setup}
\label{sec:experimental-setup}
Five identical quadrotors were constructed from off-the-shelf hobbyist components and 3D-printed parts (Fig. \ref{fig:quadrotor_system}A). The frame measures 33cm diagonally between each pair of 920kV brushless DC motors controlled by 600Hz Electronic Speed Controllers (ESCs) spinning 8"$\times$4.5" nylon propellers. A 3S 3000mAh Lithium Polymer battery provides power, and a BeagleBone Blue runs a modified version of the ArduPilot (APM) open-source software. 
A propeller guard frame consists of four custom 3D-printed corner pieces connected by eight carbon fiber rods. The propeller guard assembly measures 56cm diagonally ($\epsilon=28$cm), has a mass of 200g, and provides resilience against minor in-flight collisions. Each vehicle has mass 1075$\pm$10g.

Fig. \ref{fig:full_system} describes the experimental setup. Vehicle state is provided by an Optitrack motion capture system with eight Prime13 cameras. A Ground Control Station receives pose estimates of all vehicles from the motion capture system and uses individual 2.4GHz XBee wireless serial radios to send each vehicle its position ($\mathbf{r}_{i}, \forall i \in \mathcal{V}$), followers' desired positions ($\mathbf{r}_{i,d}, \forall i \in \mathcal{V}_F$), and a synchronization byte (Sync) all at 60Hz. A 40ms time delay was measured experimentally between Optitrack pose receipt on the vehicle and APM's onboard state estimate. 
A pilot uses a 2.4GHz DSMX Transmitter for manual override in case of system anomalies; a separate computer receives telemetry data from each vehicle over WiFi. 
All multi-vehicle tests were conducted in M-Air, a 80'$\times$120'$\times$50' outdoor netted facility at the University of Michigan. A wind vane and cup anemometer were used to obtain wind data. Outdoor tests were conducted at night to improve motion capture performance.


\begin{figure}[]
 	\centering
	\includegraphics[width = \columnwidth, trim=79 685 208 5, clip]{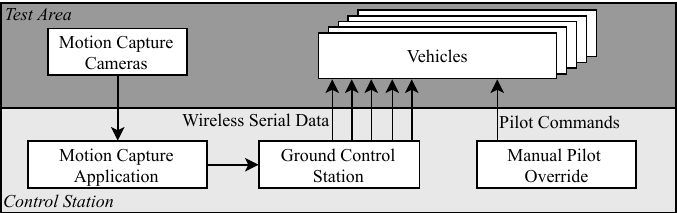}
	\caption{Experimental system block diagram} 
	\label{fig:full_system}
\end{figure}

Fig. \ref{fig:controller} outlines vehicle software. We modified APM's cascaded proportional-integral-derivative (PID) position control to use Optitrack position data. In cascaded PID,  position tracking error is scaled by a proportional gain and added to the feed-forward velocity set by  desired velocity $\dot{\mathbf{r}}_{d}$. Velocity estimates are obtained from a derivative filter on position estimate (Eq. \ref{eq:derivativeFilter}) where $T$ is the sampling period \cite{snrd}. Velocity tracking error becomes a desired acceleration and error from the acceleration tracking loop is fed into an attitude controller that relies on three-axis onboard inertial measurement unit (IMU) data. APM runs the   control loop at 400Hz. 
Leaders use Sync and a pre-loaded flight path file to generate desired positions and velocities for the position controller. Followers receive a position setpoint via wireless serial that is the weighted sum of its neighbors' real-time positions set per continuum deformation. 
Note that for a real application there would be no ground station, agents would send neighbors their positions directly, and followers would compute their desired positions on-board. While our implementation uses a centralized ground station, the actual computation of follower desired positions behaves as if there was only local communication.




\begin{equation}
    \label{eq:derivativeFilter}
    \hat{\dot{\mathbf{r}}}\left( t \right) = 
    \frac{  
        2 \left[
            \mathbf{r} \left( t - T \right) 
            - \mathbf{r} \left( t - 3T \right) 
        \right] 
        + \left[
            \mathbf{r} \left( t \right)
            - \mathbf{r} \left( t - 4T \right)
        \right] }
    {8T}
\end{equation}

\begin{figure}[]
 	\centering
 	\includegraphics[width = \columnwidth, trim=0 0 0 0, clip]{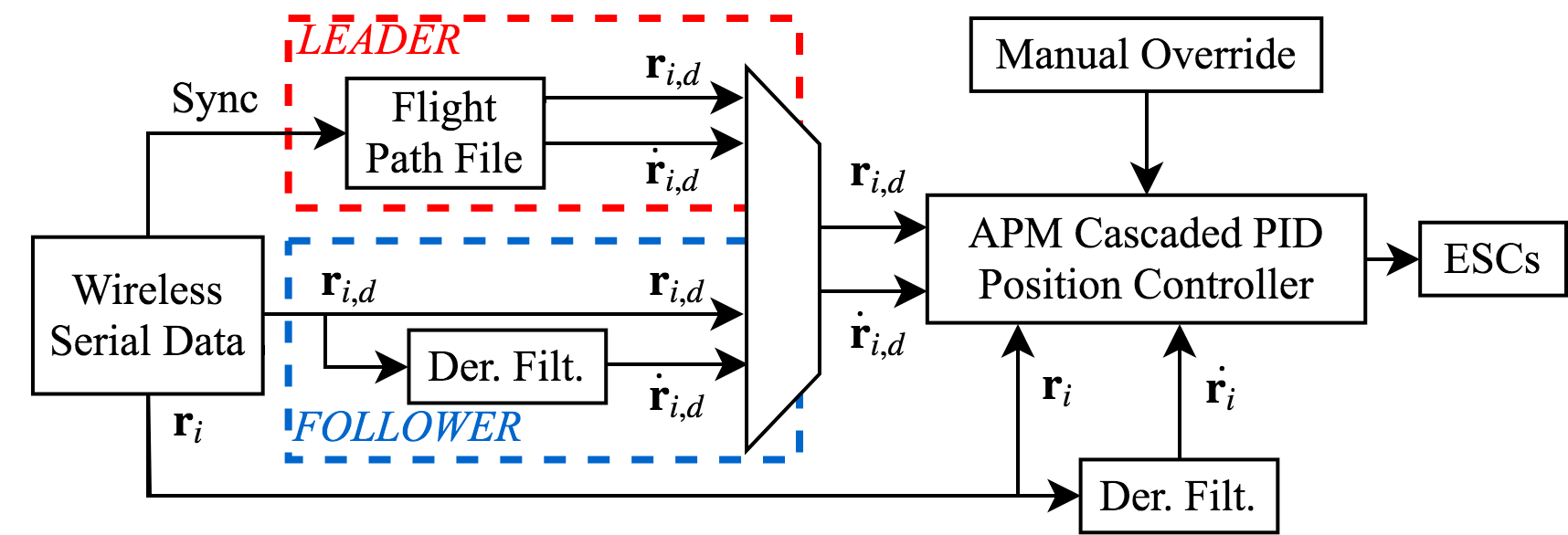}
	\caption{Vehicle software block diagram} 
	\label{fig:controller}
\end{figure}

\section{Experimental Results}
\label{sec:experimental-results}
 We conducted a series of flights with a single vehicle to empirically determine the \textit{local} deviation bound, $\delta_{l}$, where $\|\mathbf{r}_i-\mathbf{r}_{i,d}\|\leq\delta_{l}, \forall i\in \mathcal{V} $
 (Sec. \ref{subsec:controller-error}). Results were used to set bounding parameter $\delta_{g}$ as $\delta_{l}$ in the computation for $\lambda_{\mathrm{min}}$ in Eq. \ref{eqn:lamMin}. 
Given the set of prescribed leader trajectories, we command the follower agents in two ways. In the first test series, (Sec. \ref{subsec:globalDesiredResults}) we set the desired position of each follower ($\mathbf{r}_{i,d}$) to its global desired position ($\mathbf{r}_{i, HT}$) $\forall i \in \mathcal{V}_F$ as depicted in Fig. \ref{fig:bothArr}. In effect, the desired position of the followers is a weighted sum of leaders' desired positions. 
With this approach, $\delta_{l}$ and $\delta_{g}$ are equivalent, meaning collision avoidance and containment are guaranteed from Thm. \ref{thm:Thm1}. In the second approach (Sec. \ref{subsec:localCommsResults}) we compute desired follower position as a weighted sum of the actual position of three neighbor agents as prescribed in Eq. \ref{eqn:localDesired} (using the weights in Eq. \ref{eq:commWeights}). This approach, although more practical for a large and spatially distributed system (due to shorter communication links), no longer guarantees that the local deviation bound for each agent $\delta_{l}$ is a bound on the global deviation of the followers $\|\mathbf{r}_i-\mathbf{r}_{i,HT}\|$ $\forall i \in \mathcal{V}_F$. 


\begin{equation}
\begin{bmatrix}
w_{4,1}\\
w_{4,3}\\
w_{4,5}
\end{bmatrix}
=
\begin{bmatrix}
0.5\\
0.134\\
0.366
\end{bmatrix}
,
\begin{bmatrix}
w_{5,2}\\
w_{5,3}\\
w_{5,4}
\end{bmatrix}
=
\begin{bmatrix}
0.5\\
0.134\\
0.366
\end{bmatrix}
\label{eq:commWeights}
\end{equation}


\subsection{Determining Local Deviation Bound}
\label{subsec:controller-error}


To determine local deviation bound $\delta_{l}$, we flew a single vehicle in a 1m edge-length square while varying parameters to characterize error sources. Parameters varied include battery voltage and environment (indoor vs. outdoor with negligible wind), largest desired velocity from quintic spline guidance ($v_\mathrm{max}$), and altitude above ground level. We also considered the effect of disturbances induced by nearby  vehicles by flying all five agents, with followers using global desired positions, to see if controller performance varied from the single-vehicle tests. 

Results show that battery voltage, indoor vs. outdoor environment, and altitude above ground (minimum 1m) have negligible impact on $\delta_l$. Battery voltage was varied from 11.4V to 12.5V, a typical operating range, with less than 10\% change in $\delta_{l}$. 
We varied flight altitude from 0.75m to 1.75m in increments of 0.25m, with a 15\% larger standard deviation in error at 0.75m than at other altitudes. Since all altitudes above 0.75m had similar performance, we flew multi-agent tests at 1.5m above ground level.

To check dependence on $v_\mathrm{max}$, a single vehicle was flown using the quintic spline guidance law from Eq. \ref{eqn:spline}. 
We relax the constraint on final desired position $\mathbf{r}_{i,d}(t_f) = \mathbf{r}_{i,f}$ and instead command a particular velocity halfway through each flight segment $\dot{\mathbf{r}}_{i,d}(\frac{t_f - t_0}{2}) = v_{max}$. For this test series, time of flight between each waypoint pair is fixed at ($t_f - t_0$) = 3.75s which results in a translation of 1m when $v_\mathrm{max}$ = 50$\frac{cm}{s}$.



     
\begin{table}[]
    \centering
    \begin{tabular}[pos]{r || c | c | c}
     $v_\mathrm{max}$ (cm/s) & Mean (cm) & Std. Dev (cm) &  Max (cm) \\ \hline
     0 & 4.80 & 1.90 & 12.44 \\
     25 & 6.64 & 3.04 & 15.95 \\
     50 & 7.94 & 4.27 & 22.65  \\
     75 & 11.64 & 6.65 & 29.77  \\ 
     100 & 14.53 & 9.23 & 36.44  \\ 
\end{tabular}
    \caption{Statistical parameters of local deviation increase as $v_\mathrm{max}$ increases. Tests performed indoors at 150cm altitude.}
    \label{tab:errVmax}
\end{table}

Table \ref{tab:errVmax} shows that tracking error mean, max, and standard deviation increase as $v_\mathrm{max}$ increases. We believe a large source of error is introduced in state estimate delay. Each agent receives a position estimate from Optitrack with about a 40ms delay. Then, a derivative filter (Eq. \ref{eq:derivativeFilter}) is used on position data to estimate the actual velocity. The filter adds an additional delay to the velocity estimate, but characterization of the derivatives filter's delay and its effects on position error was not explored for this paper.  



For data shown in Table \ref{tab:errVmax}, the $v_\mathrm{max} = 0\frac{cm}{s}$ hover case is computed over 40 flights. The other datapoints, $v_\mathrm{max} = \{25, 50, 75, 100\} \frac{cm}{s}$ in Table \ref{tab:errVmax}, have ten flights for each $v_\mathrm{max}$. All flights in this dataset were conducted indoors at an altitude of 150 cm above the floor. 

To investigate downwash impact on neighboring agents, we flew all five quadrotors with followers using global desired positions and compared results to those in Table \ref{tab:errVmax}. Leaders flew the trajectory from  Sec. \ref{sec:generating-path} with $\delta_{g}$ = 40cm, $(t_f - t_0)$ = 3.75s and $v_{max}$ = 50$\frac{cm}{s}$ as before. Two full flights with five vehicles were flown  resulting in ten datasets overall.

\begin{table}[]
    \centering
    \begin{tabular}[pos]{c || c | c | c}
     & Mean (cm) & Std. Dev (cm) &  Max (cm) \\ \hline
    Five-Agent Flights &  9.69 & 5.16 & 30.33\\
    Single-Agent Flights &  7.94 & 4.27 & 22.65 \\
\end{tabular}
    \caption{Statistical parameters of inter-agent disturbance characterization. Tests performed outdoors at 150cm altitude}
    \label{tab:interAgentDisturbance}
\end{table}

Table \ref{tab:interAgentDisturbance} shows the results of the five-agent flights alongside the single-agent flight data copied from the 50$\frac{cm}{s}$ case in Table \ref{tab:errVmax}. For the five-agent flights, agent 4 had an anomaly where it did not run any controller-related computations for 0.3s resulting a maximum error of 41.68cm. which is treated as an outlier and excluded from error characterization. 

During the five-agent flights, the XBees occasionally drop packets. Fig. \ref{fig:dt} shows the change in time $\Delta t$ between two consecutive Optitrack pose estimates received by an agent over Xbee. Fig. \ref{fig:dtGood} shows a typical single-agent flight while Fig. \ref{fig:dtBad} is from a five-agent flight where the larger $\Delta t$ corresponds to an integer number of packets lost. 
Five-agent flights may have larger error in part because of Optitrack data packet drop. 
Results from 
Tables \ref{tab:errVmax} and \ref{tab:interAgentDisturbance} show that a local deviation bound of $\delta_{l}$ = 40cm is sufficient. The first five-agent flights sent followers global desired positions to support analysis of continuum deformation constraints in Sec. \ref{subsec:globalDesiredResults}. 

\begin{figure}[]
 	\centering
 	\begin{subfigure}[c]{0.48\columnwidth}
	    \includegraphics[width = \textwidth,trim=100 230 130 230, clip]{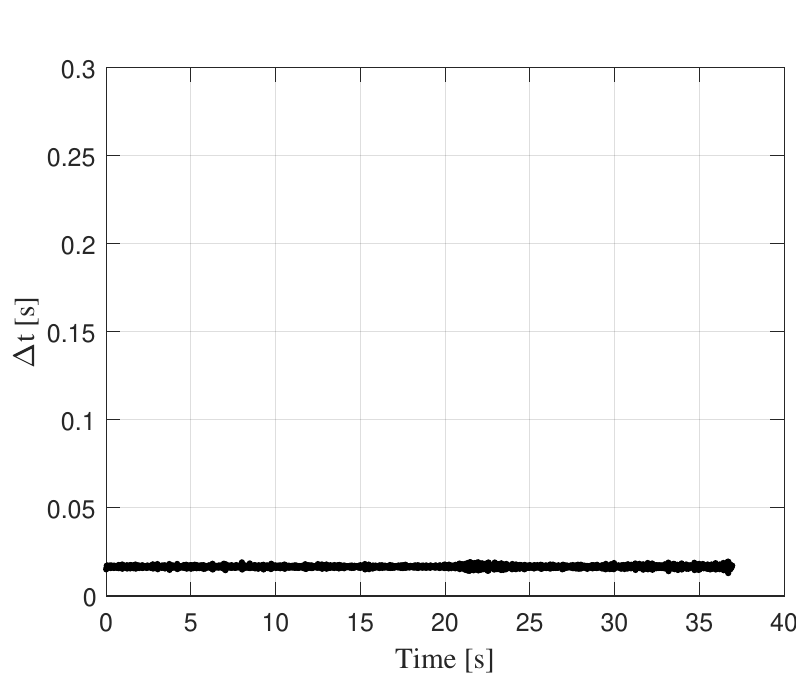}
	    \caption{Single-agent flight.} 
	    \label{fig:dtGood}
    \end{subfigure}
  	\begin{subfigure}[c]{0.48\columnwidth}
	    \includegraphics[width = \textwidth,trim=100 230 120 230, clip]{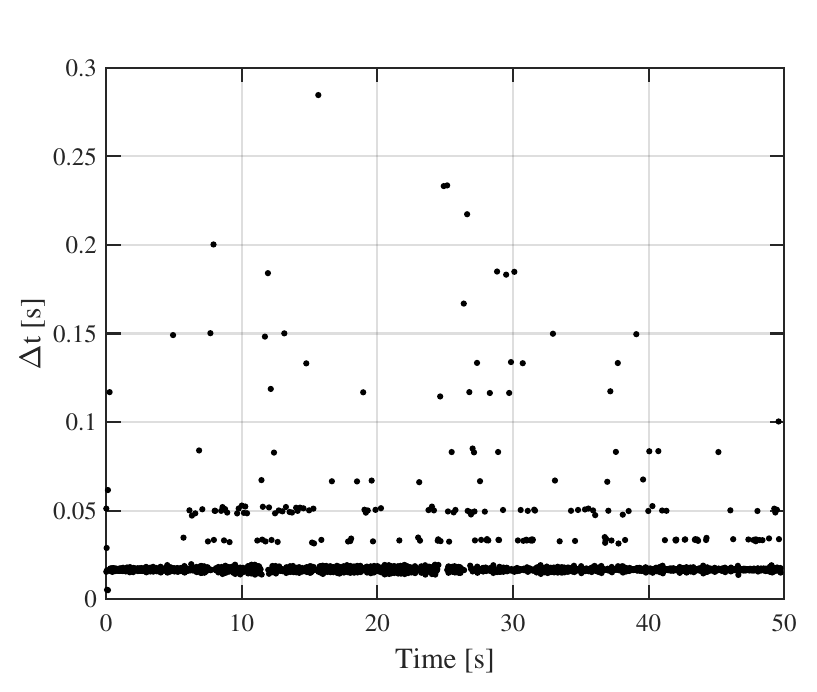}
	    \caption{Five-agent flight.} 
	    \label{fig:dtBad}
    \end{subfigure}
	\caption{Data is received consistently over Xbee radios at 60Hz for single-agent flights. Packets dropped in five-agent flights.} 
	\label{fig:dt}
\end{figure}
\subsection{Five-Agent Flight: Followers with Pre-set Waypoints}
\label{subsec:globalDesiredResults}
Fig. \ref{fig:expTest1} shows the results for a five-agent test with follower positions based on a weighted sum of the leader desired positions. The wind speed was measured to be 0mph. $\delta_{g}$ = 40cm was used as a conservative error bound. This test was run to distinguish experimental setup error sources (downwash, communication interference, etc) from coordination-induced error. The four plots in Fig. \ref{fig:expTest1} show team performance with respect to the four constraints from Sec. \ref{sec:generating-path}: follower containment within the leading triangle, collision avoidance between agents, bounding of deviation from the local desired position, and bounding of deviation from the global desired position. The black vertical lines on each plot use the same symbols from Fig. \ref{fig:trajOutline} to show the desired waypoint of the leaders in time. The unlabeled vertical lines between Star ($\star$) and Diamond ($\lozenge$) represent an intermediate waypoint added to fit all agents within our motion capture workspace. 

Fig. \ref{fig:expTest1}A shows the distance from both followers to the leader boundary with a flat horizontal line at $\epsilon$ = 28cm and another horizontal line at $-\delta_{g}$ = -40cm. Crossing the $\epsilon$ line means the follower left the leader boundary while crossing the $-\delta_{g}$ line means the follower left the outer bounding triangle. Thus, the followers never leave the leading triangle since their lines never go below the horizontal $\epsilon$ line. Fig. \ref{fig:expTest1}B shows each agent's distance to its nearest neighbor along with a $2\epsilon$ = 56cm line. There are no inter-agent collisions because no agent has a distance that goes below the $2\epsilon$ line. Fig. \ref{fig:expTest1}C and Fig. \ref{fig:expTest1}D show the deviation of each agent from its local and global desired position. No agent has a local or global deviation exceeding $\delta_{g}$ = 40cm (represented by the horizontal line). Figs. \ref{fig:expTest1}C and \ref{fig:expTest1}D are identical since followers are being commanded to global desired positions. 



\begin{figure}[]
 	\centering
	\includegraphics[width = \columnwidth,trim=0 135 20 140, clip]{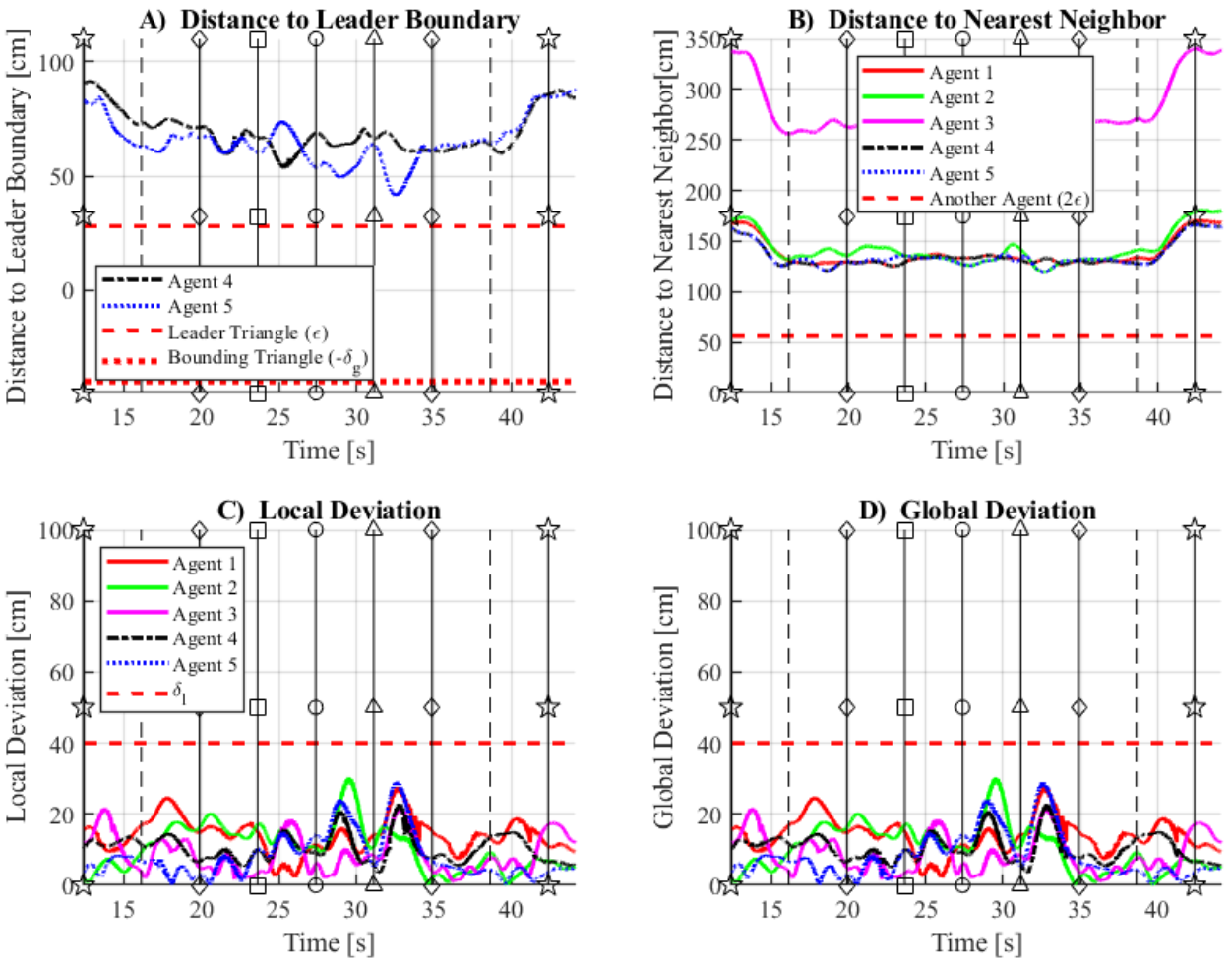}
	\caption{In-flight constraint values over time for the five-quadrotor continuum deformation trajectory.  Follower desired positions are computed from desired leader positions.} 
	\label{fig:expTest1}
\end{figure}


\subsection{Five-Agent Flight: Followers with Local Communication }
\label{subsec:localCommsResults}
Fig. \ref{fig:expTest3} shows results for the five quadrotor system running continuum deformation under local communication. The average wind speed was 2.1mph at 75$^{\circ}$ relative to +X (effectively the direction of motion in Fig. \ref{fig:trajOutline}C: $\square$ to $\bigcirc$).   

Figs. \ref{fig:expTest3}B and Figs. \ref{fig:expTest3}C show there are no inter-agent collisions and that no agent exceeds the local deviation bound. Fig. \ref{fig:expTest3}D shows that followers violate the global deviation constraint in all segments except when the formation moves in the direction of the wind:($\square$ to $\bigcirc$) and ($\lozenge$ to the intermediate waypoint). In segments where the global deviation constraint is satisfied, neither follower leaves the leader triangle as predicted by Thm. \ref{thm:Thm1}. However, there is a segment (intermediate waypoint to $\lozenge$) where followers violate global deviation while within the leader boundary.

The large follower deviation in Fig. \ref{fig:expTest3}D is caused in part by steady-state error in agent y-components, most likely due to wind which effectively shifted the entire formation in +Y. 
Follower local desired positions ($\mathbf{r}_{i,d}$) shift with the leaders while their global desired positions ($\mathbf{r}_{i, HT}$) are unaffected, resulting in a compounding error effect with respect to the follower global desired position reference.

\begin{figure}[]
 	\centering
	\includegraphics[width = \columnwidth,
	trim=0 135 20 140,
	clip]{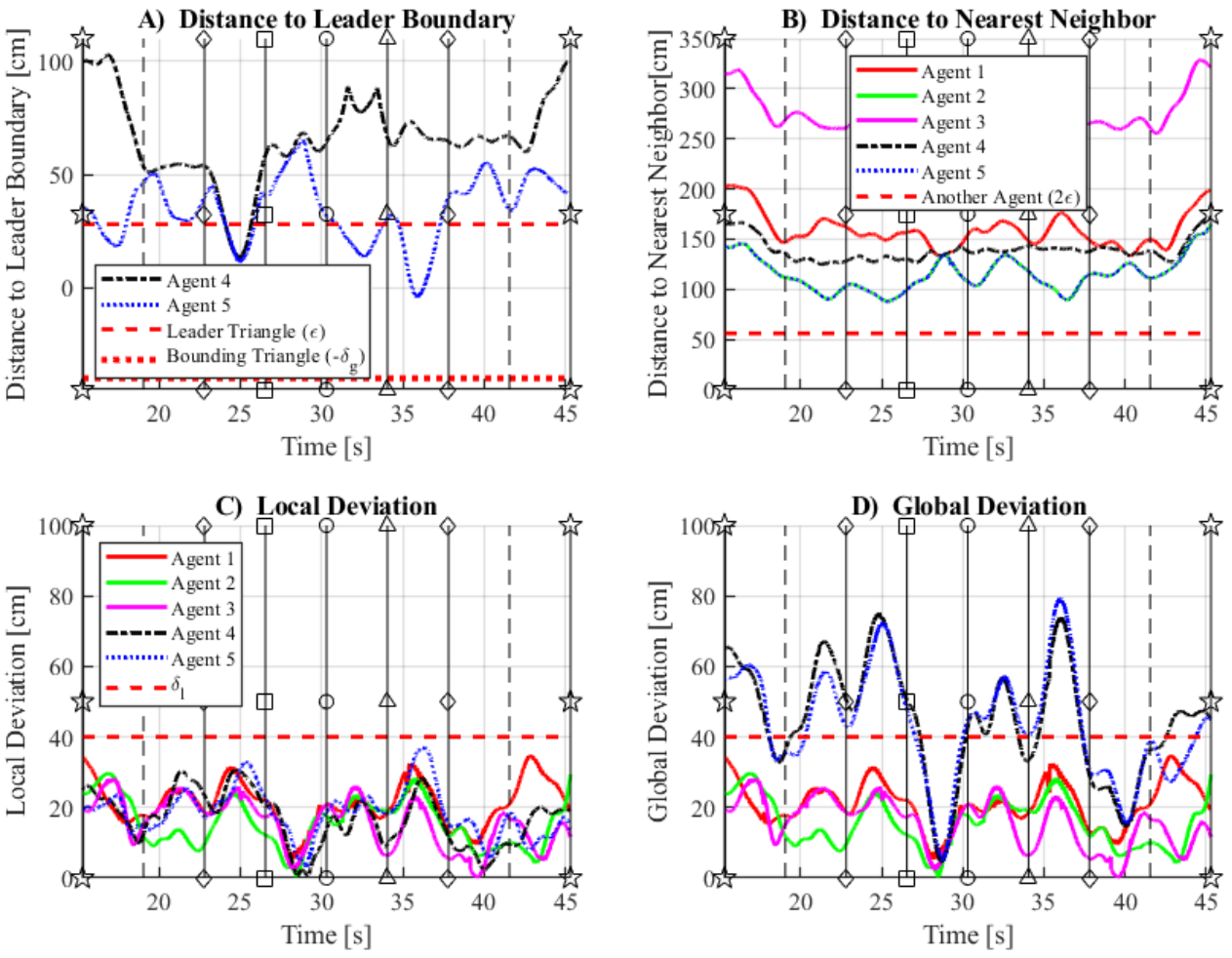}
	\caption{In-flight constraint values over time for the five-quadrotor continuum deformation trajectory.  Follower desired positions are computed from local communication of actual neighbor positions.}
	\label{fig:expTest3}
\end{figure}



\section{Discussion}
\label{sec:discussion}

In tests where local deviation is equivalent to global deviation (Fig. \ref{fig:expTest1}), all four constraints were met. Under strictly local communication of actual positions (Fig. \ref{fig:expTest3}), collision and local deviation constraints were met while containment and global deviation were violated. All constraints were satisfied in segments of the flight where the global deviation constraint was met (Thm. \ref{thm:Thm1}). However, we found it difficult to predict follower deviation in advance as time delay and error compounding from leaders to follower could not be observed in single-vehicle tests. A large error bound can be prescribed to satisfy all constraints based on evaluation of errors observed in continuum deformation flights, but such a conservative approach would require the quadrotors to be substantially separated which may not be practical.

\section{Conclusion}
\label{sec:conclusion}
This paper presented results from the first experimental evaluation of continuum deformation using a five-quadrotor team. 
Results successfully demonstrated the theory in an outdoor motion capture environment but motivate follow-on work to incorporate knowledge of single-vehicle error bounds directly into the design of safe leader trajectories and in turn safe continuum reference geometries.
Although this paper demonstrated scaling and translation of a 2D leading triangle, continuum deformation allows shear, rotation, and composition of these four fundamental deformations simultaneously. We plan to conduct further experiments that test constraint satisfaction given more complex leader trajectories.

\section*{Acknowledgements}
This work was supported in part by National Science Foundation (NSF) Grant CNS 1739525.


\bibliographystyle{IEEEtran}
\bibliography{IEEEabrv,refs}

\end{document}